\documentclass[11pt]{article}

\usepackage[margin=1in]{geometry}
\usepackage{amsmath,amssymb,amsthm}
\usepackage{graphicx}
\usepackage{enumitem}
\usepackage[round]{natbib}
\usepackage[colorlinks=true,linkcolor=blue,urlcolor=blue,citecolor=blue]{hyperref}

\usepackage{tikz}
\usetikzlibrary{arrows.meta, positioning, calc, fit, backgrounds}
\tikzset{
  txn/.style={circle, draw=black!70, fill=black!5, minimum size=8.5mm, inner sep=1pt, font=\small},
  multi/.style={draw=black, very thick},
  uedge/.style={-{Stealth[length=2.6mm]}, semithick},
  aone/.style={uedge, draw=blue!65!black},
  atwo/.style={uedge, draw=orange!85!black},
  aboth/.style={uedge, draw=violet!85!black, thick},
  cyc/.style={uedge, draw=red!75!black, thick},
  imposed/.style={-{Stealth[length=2.6mm]}, densely dashed, draw=teal!65!black, semithick},
  unranked/.style={{Stealth[length=2mm]}-{Stealth[length=2mm]}, densely dotted, draw=gray!90},
  lab/.style={font=\scriptsize, fill=white, inner sep=1.5pt},
  grp/.style={draw=gray!70, densely dashed, rounded corners=6pt, inner sep=9pt}
}

\theoremstyle{plain}
\newtheorem{lemma}{Lemma}
\newtheorem{proposition}{Proposition}

\theoremstyle{definition}
\newtheorem*{defn}{Definition}
\newtheorem*{exmpl}{Example}

\theoremstyle{remark}

\title{Ordering by Unanimity: Giving Applications Sequencing Rights Without Breaking Composability\thanks{I am grateful to Kushal Babel, Jan Camenisch, Andrei Constantinescu, Lioba Heimbach, Jason Milionis, Mike Setrin, and participants in the Category Labs internal research seminar for their comments and suggestions.}}
\author{Andrea Canidio \\ Category Labs}
\date{\today}

\begin{document}

\maketitle

% ===== Revised front matter: abstract + introduction + literature review =====
% Paste this over the current front matter -- everything from the abstract through
% the end of the Literature paragraph (between the title and the model section).
% First person throughout; resolved editorial notes removed.

\begin{abstract}
Blockchain applications may have preferences over the order in which transactions execute: an automated market maker may use an external price feed to price its liquidity, and require that the oracle update incorporating this price execute before any swap; an order-book exchange may want to execute cancellations of limit orders before incoming market orders; an application may run an on-chain auction by executing bids from highest to lowest, so that the first bid wins. The problem is that, currently, the ordering of transactions is chosen by the underlying blockchain and may not be compatible with the requirements of a specific application. In this paper, I tackle this problem by introducing an algorithm called \textit{unanimity override}. The intuition is that when all the applications agree on how to order two transactions, the underlying blockchain should respect this agreement; a default order --- the order in which transactions appear in the block --- settles the rest. The problem with this naive approach is that application unanimity is intransitive and may form cycles, which the algorithm must break. Cycle-breaking is also the rule's main vulnerability because an attacker can insert transactions to manufacture a cycle. Yet two main guarantees hold against any attacker who sets the default order, deploys applications, and inserts transactions. First, all transactions that interact with a single application that expressed preferences are ordered according to that application's preferences, even when they also interact with other applications that did not express preferences. Second, \textit{gated} transactions --- those that cannot be outranked in the unanimity order by any transaction crafted by an attacker --- always execute as the applications unanimously prefer, even when they touch many applications. These results identify exactly which preferences a protocol can protect, and they tell applications and senders in advance which transactions will execute in the intended order.
\end{abstract}

\section{Introduction}
 
A blockchain is a single, shared ledger on which multiple applications run and interact. This integration among applications on the same chain is called \textit{composability}, and it is one of the main economic benefits of blockchains. For example, when a new asset is issued on chain --- a native cryptocurrency, or an on-chain version of a traditional asset such as a stock or bond --- anyone can then make it tradable on an existing market or usable as collateral in an existing lending protocol. Doing so requires a few steps, but no one's permission. Applications also build on one another: a lending protocol can liquidate collateral through an existing market, again, without needing an explicit integration or agreement from this market. And a single transaction can touch several applications at once --- passing through an on-chain router to trade across markets, or trading on two markets at once to arbitrage away a price difference. Composability holds the promise of a globally integrated, decentralized financial system built on blockchains. 
 
At the same time, the fact that all applications run on the same blockchain also limits the applications that can be deployed to those compatible with the blockchain's underlying design. This limitation is most evident in how transactions are ordered. A blockchain's participants --- validators, nodes, and users --- must agree not only on which transactions to include in a block, but also on the order in which they appear: transactions execute in the order the block lists them, and reordering them changes the execution outcome. For this reason, the right to order transactions is usually allocated to a block \emph{proposer} or \emph{sequencer}, who may order them at will or follow a fixed, protocol-defined rule such as priority-fee ordering.
 
Allocating ordering rights to the underlying blockchain, however, constrains the application design space, and this constraint appears to be binding. For example, a PropAMM --- an automated market maker (AMM) that takes its price from an oracle --- needs that oracle's update to execute before any swap against it. PropAMMs emerged on Solana because this requirement was compatible with Solana's default ordering rule, but they cannot be easily implemented on other chains. An exchange such as Hyperliquid gives execution priority to order cancellations, and created its own chain to do so. Rollups are off-chain mechanisms that aggregate and order transactions. Like independent chains, rollups acquire ordering rights at the cost of isolation. Finally, several decentralized-finance designs have stalled or failed because they require a specific order to function: an AMM that auctions the right to rebalance its pool by ranking swaps on the fee they pay, or an on-chain auction that must process bids from highest to lowest. 

 In this paper, I propose an ordering rule that lets applications express preferences over how transactions are ordered, and respects those preferences as far as possible without sacrificing safety or composability. The rule, which I call \emph{unanimity override}, is easy to state: whenever the applications with a stake in a pair of transactions agree on how to order them, the protocol follows that agreement, unless the resulting preferences form a cycle, in which case a fallback breaks the cycle. A default order (i.e., the order in which transactions appear in the block) is then used to recover a complete execution order.
 
 The key to the mechanism is how it breaks cycles. I show that cycles can arise only when at least two transactions interact with multiple \emph{opinionated} applications --- applications that have expressed ordering preferences. Breaking a cycle then means \emph{demoting} one or more of these transactions. The protocol removes a demoted transaction's unanimity comparisons with the others in its cycle and places it after them in the execution order; when several transactions are demoted, the default order ranks them. This breaks the cycle, but overrides the applications' preferences for the demoted transactions. This, in turn, can be exploited by an attacker, who can insert transactions to create a cycle and manipulate the execution order.
 
The main results identify which of an application's preferences are protected. Two main guarantees hold even against an attacker who controls the default order, deploys applications, and injects transactions. First, transactions that interact with only one opinionated application are executed according to that application's preferences. The reason is that, when both transactions interact with a given opinionated application and at least one of them is single-opinion, only that application can express preferences over the ranking of these transactions. Hence, the unanimity ranking coincides with that application's preferences. Because the fallback never demotes single-opinion transactions, the execution order always enforces the unanimity ranking between single-opinion transactions, which coincides with the relevant application's preferences. Second, a transaction that no attacker-crafted transaction can outrank in the unanimity ranking cannot lie on a cycle, so the fallback never demotes it; it executes in the order the applications unanimously prefer, no matter how many applications it interacts with. A symmetric guarantee
holds at the bottom of the order: a transaction that always ranks below any
attacker-crafted transaction in the unanimity ranking never lies on a
cycle, so every unanimous comparison placing it below another transaction is enforced --- a settlement transaction that every application ranks last
is the natural example.

I then argue that these guarantees cannot be improved upon, at least ex
ante --- before the block is formed. The transactions the guarantees leave
out interact with multiple opinionated applications and are gated neither
above nor below. These are exactly the transactions over which applications
can disagree. To see this, take a transaction $x$ not covered by the above guarantees: it interacts with at least two opinionated application; one
application it interacts with ranks some transaction $y_1$ strictly above
$x$; a different application it interacts with ranks some transaction
$y_2$ strictly below $x$. A single transaction $y$ can then bundle the
interactions of $y_1$ and $y_2$. If the block contains both $x$ and $y$,
the two applications disagree on their order: one wants $y$ to execute
first, the other wants $x$. No ordering rule, whatever its form, can
satisfy both --- so none can guarantee every application's preferences over these transactions.

Returning to the applications that motivated the rule, oracle updates and cancellations --- including cancellations that span many markets --- are protected by the second guarantee: no transaction an attacker can craft ever ranks strictly above them in the relevant applications' rankings, so they cannot be attacked. Auctions, and AMMs that rank swaps by the fee they pay, rest on the first guarantee together with the sender's incentive: a sender keeps a transaction safe by routing it through a single opinionated application. Because the execution order of single-opinion transactions always follows the relevant application's preferences, unanimity override preserves --- and extends --- the guarantees a rollup provides. In the case of a rollup, transactions are, by assumption, confined to interacting only with the rollup --- they cannot touch other applications deployed on the same chain. Unanimity override also orders these transactions according to the rollup's preferences. But it extends the guarantee to transactions that also interact with non-opinionated applications on the same chain.\footnote{A case in point is the bridging of assets between a blockchain and its rollup. Non-native assets (i.e., assets other than ETH on Ethereum) are represented on chain by a smart contract --- an application, in the language of this paper --- with no ordering preferences. Because a rollup cannot interact directly with these contracts, assets must be bridged: a third party locks them on the base chain and ``reissues'' them on the rollup. Bridges have been a constant source of security concerns, and several have been hacked over the years. Under unanimity override, an application with rollup-like sequencing guarantees is deployed directly on the base chain, where its transactions can also touch the asset contracts: no bridge is needed.} And it allows the rollup's transactions to interact with other opinionated applications, although in that case the ordering may not be guaranteed.

I then compare the guarantees of unanimity override against two benchmarks: one where the attacker freely chooses the execution order, and one where the order is determined by the priority fee each transaction pays. 
Against an attacker who can order transactions freely, unanimity override strictly reduces what the attacker can do. Against priority-fee ordering, the comparison is less clear-cut. Unanimity override protects the transactions the guarantees cover, but it can also make some other transactions cheaper to attack, and it enables griefing attacks that priority-fee ordering rules out. 

A final consideration is that users, and not only applications, may have
ordering preferences. When an application declares its preferences and a user
submits a transaction to it, I treat the application's preferences as standing
in for the user's: by respecting the former, the protocol also respects the
latter. The converse need not hold. An application may express no
preference --- a traditional AMM, for example --- while its users still have
one, typically not to be front-run. Whether unanimity override protects these
transactions depends on how the repaired unanimity relation is completed into
a full execution order. I show that when the completion uses Kahn's
topological-sort algorithm, these transactions keep the benchmark's
protection: front-running them is never cheaper under unanimity override than
under priority-fee ordering.
 
Finally, an important caveat is that the paper assumes throughout a strong form of censorship resistance. Without it, the proposer could manipulate the order simply by withholding transactions. The analysis is therefore complementary to work on multiple concurrent proposers, transaction encryption, and inclusion lists, all of which aim to provide censorship resistance (see, for example, \citealp{avarikioti2023fnf}, \citealp{garimidi2025multipleconcurrentproposers}, \citealp{cryptoeprint:2025/194}, \citealp{babel2026cadenceextremepipeliningmultiple}).

The rest of the paper proceeds as follows. The remainder of this section reviews the related literature. The following section presents the model: a block of transactions ordered by a default rule, and applications that rank the transactions they interact with. Section~\ref{sec:override} introduces the unanimity-override rule and the fallback that resolves cycles. Section~\ref{sec:threat-model} presents the threat model and derives the rule's guarantees. Section~\ref{sec:impossibility} shows that these guarantees cannot be improved upon: no ordering rule, whatever its form, can protect the transactions they leave out. Section~\ref{sec:no-opinion} shows that the rule also protects the transactions over which applications express no preferences: front-running them requires outbidding them, exactly as under priority-fee ordering. Section~\ref{sec:robustness} generalizes the results to any aggregation rule that is Paretian, provided again that cycles are broken according to a specific logic. Section~\ref{sec:interaction-sets} discusses a practical issue: how to determine which applications a transaction interacts with --- and therefore which applications can express preferences over it. The final section concludes.
 
\paragraph{Literature.} To my knowledge, the only paper that directly takes up the question of how to delegate ordering rights to applications is \citet{durvasula2026monotoneprioritysystemfoundations}. It is closely related, but its question is different: there, contract
developers attach integer priorities to their function calls, subject to a
monotonicity restriction that makes all declared constraints jointly
enforceable. Here,
applications rank transactions independently, and their preferences may
not be jointly enforceable; the goal is to identify the guarantees the protocol can provide
nonetheless. The expressivity also differs: priorities attached to function
calls are enough to put cancellations or oracle updates at the top of a
block, but not to rank competing swaps by how much they pay. The paper also discusses, with useful references, the debate around application-specific sequencing \citep{ferreira2022credible} --- an \emph{off-protocol} approach, in which transactions are sequenced outside the base protocol. That approach is related to mine but distinct from it: I study a rule internal to the protocol.
 
The aggregation of preferences into an execution order is an exercise in social choice.  The link between transaction ordering and social choice has already been noted \citep{ramseyer2024brief}, especially by the literature on \emph{order fairness} \citep{kelkar2020aequitas,kelkar2023themis}. Those papers aim to aggregate the orders in which validators receive transactions, and they also have to deal with cycles. In both settings, an attacker can try to disrupt the protocol by creating cycles \citep{vafadar2023condorcet}.
 
The goal of this paper is nonetheless very different from that of the literature on fair ordering. Order fairness enforces a protocol-determined notion of fairness; the goal of this paper is to expand the application design space, giving applications control over how their own transactions are ordered. The main results are therefore the guarantees that unanimity override provides to applications over the ordering of transactions. Because the two tackle different problems, unanimity override and fair ordering can coexist: a block could be ordered fairly and its transactions then executed in the order determined by unanimity override. I speculate that the two are complementary: I derive the guarantees in an adversarial environment where the attacker can manipulate the default order at will; those guarantees can only strengthen when order fairness limits the attacker's degrees of freedom.
 
The non-transitivity of the unanimity ranking is a direct consequence of  Arrow's impossibility theorem \citep{arrow1963}, which states that no rule aggregating three or more alternatives can be at once Paretian, non-dictatorial, independent of irrelevant alternatives, and guaranteed to return a transitive order. My composite rule --- unanimity followed by the fallback --- is Paretian, non-dictatorial, and always returns a complete transitive order. It escapes the theorem by giving up the remaining axiom, independence of irrelevant alternatives: when there is a cycle, the order between two transactions can depend on which other transactions share the cycle and on the default order, not on the two alone.

The literature following \citet{arrow1963} has studied various ways to deal with the impossibility theorem. However, the problem I consider differs from those studied in that literature in two ways. First, in the case considered here, the goal is to recover a full execution order, not just a choice function. Second, the effective electorate varies across pairs --- each pair is judged only by the applications that touch both transactions --- so the aggregate relation can contain cycles. The first departure makes a classical response to Arrow's theorem, the \emph{Pareto-extension rule}, inapplicable here. Introduced by \cite{sen1969quasitransitivity,sen1970collective}, the Pareto-extension rule completes the Pareto relation by declaring every Pareto-incomparable pair socially indifferent; it keeps the Pareto principle, independence of irrelevant alternatives, and non-dictatorship, at the cost of full transitivity, since the resulting relation is only quasi-transitive. Quasi-transitivity suffices to define a social choice --- a most-preferred option among a set always exists --- but not to produce a complete order. The second departure connects to the theory of \emph{ordering extensions}, which studies how to complete a binary relation into a full order: \citet{szpilrajn1930extension} shows that any quasi-order extends to a complete order, and \citet{suzumura1976remarks} weakens the requirement to Suzumura consistency --- the absence of cycles containing a strict preference. These results do not apply directly here, because aggregating the preferences of applications with overlapping domains can generate exactly the strict cycles that Suzumura consistency rules out.

\section{The model}

There is a block of transactions $X=\{x_1,x_2,\ldots,x_n\}\subset\mathbb X$, where $\mathbb X$ is the set of all possible transactions. An injective function $U_d:X\rightarrow\mathbb R$ represents a default strict total order over $X$: if $U_d(x_i)>U_d(x_j)$, then $x_i$ is higher in the block than $x_j$. For now, I stay agnostic about who defines $U_d$ --- it may be a protocol rule, such as priority-fee ordering, or the proposer's own ranking.

There is also a set of applications $A=\{a_1,a_2,a_3,\ldots\}$. A transaction may interact with an application --- most simply, by calling one of its functions. For each application $a_t$, let
$\mathbb X_t\subseteq\mathbb X$ be the set of transactions that interact with
$a_t$, and let $X_t\equiv X\cap\mathbb X_t$ be the application's
\emph{interaction domain}: the transactions in the block that interact with
it.

Each application has preferences over the order in which the transactions
that interact with it execute, represented by a function
$U_t:\mathbb X_t\rightarrow\mathbb R$:
\begin{itemize}
  \item if $x_i,x_j\in\mathbb X_t$ and $U_t(x_i)>U_t(x_j)$, then application
  $a_t$ prefers $x_i$ to execute before $x_j$ whenever both appear in the
  same block;
  \item if $x_i,x_j\in\mathbb X_t$ and $U_t(x_i)=U_t(x_j)$, then application
  $a_t$ is indifferent to their execution order.
\end{itemize}
In practice, an application ranks a transaction by what it observes of it:
the function called, the parameters, the value sent, the sender's address,
the priority fee it pays. The model
leaves this content implicit and works directly with the preferences it
induces. Only transactions whose interaction the application accepts belong
to $\mathbb X_t$: a transaction whose interaction with $a_t$ is invalid --- for
example, because it calls a function reserved for other senders --- is treated as
not interacting with $a_t$ at all, and a transaction with valid and invalid
interactions is ranked only by its valid interactions.

In practice, the interaction domains may be difficult to determine, because
actual interactions can depend on the execution order; I return to this issue
in Section~\ref{sec:interaction-sets}. For now, I take the interaction
domains as given, for example because each transaction must declare them
beforehand (similarly to Solana's account list).

Finally, let
\[
  \mathcal O\equiv\{t:a_t\in A\text{ and }U_t\text{ is not constant on }\mathbb X_t\}
\]
be the set of \emph{opinionated applications}. These are the applications that can strictly rank at least one pair of transactions. For every possible transaction $x\in\mathbb X$, define
\[
  T^{\mathrm o}(x)\equiv\{t\in\mathcal O:x\in\mathbb X_t\}
\]
as the set of opinionated applications with which $x$ interacts. I call $x$ a \emph{multi-opinion transaction} if $|T^{\mathrm o}(x)|\geq2$, a \emph{single-opinion transaction} if $|T^{\mathrm o}(x)|=1$, and a \emph{no-opinion transaction} if $T^{\mathrm o}(x)=\varnothing$.

\section{Application unanimity override}\label{sec:override}
 
I now propose a rule determining the order of execution of transactions $\succ_e$. Such a rule is meant to capture the following observation: whenever all applications agree on a given ordering and such ordering has no cycles, the protocol should follow the applications' ordering preferences.
 
Formally, I introduce a directed unanimity relation $\succ_u$.  
 
\begin{defn}
$x_i \succ_u x_j$ if and only if:
\begin{itemize}
  \item $\forall t$ such that $x_i, x_j \in X_t$, $U_t(x_i)\geq U_t(x_j)$;
  \item $\exists t$ such that $x_i, x_j \in X_t$ such that $U_t(x_i) > U_t(x_j)$.
\end{itemize}
\end{defn}

Note that $\succ_u$ over $X$ is neither complete nor transitive. As an example, suppose there are two applications with $X=X_1=X_2=\{x_1,x_2,x_3\}$ and $U_1(x_1)>U_1(x_2)>U_1(x_3)$ but $U_2(x_1)>U_2(x_3)>U_2(x_2)$. That is, both applications agree that $x_1$ should be top of block, but disagree on the ordering of the other two transactions: $x_2$ and $x_3$ cannot be unanimously ranked against each other, so $\succ_u$ on $X$ is not complete. Figure~\ref{fig:unanimous-preferences} shows the resulting unanimity graph.
 
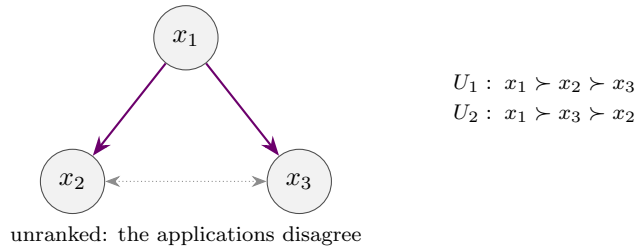
\begin{figure}[ht]
  \centering
  \begin{tikzpicture}
    \node[txn] (x1) at (0,1.9) {$x_1$};
    \node[txn] (x2) at (-1.5,0) {$x_2$};
    \node[txn] (x3) at (1.5,0) {$x_3$};
    \draw[aboth] (x1) -- (x2);
    \draw[aboth] (x1) -- (x3);
    \draw[unranked] (x2) -- node[lab, below=15pt] {unranked: the applications disagree} (x3);
    \node[font=\scriptsize, align=left, anchor=west] at (3.4,1.1)
      {$U_1:\; x_1 \succ x_2 \succ x_3$\\[2pt]
       $U_2:\; x_1 \succ x_3 \succ x_2$};
  \end{tikzpicture}
  \caption{Unanimous preference graph for the two-application example above. Solid arrows are unanimous comparisons of $\succ_u$; the dotted link marks the pair that $\succ_u$ leaves unranked.}
  \label{fig:unanimous-preferences}
\end{figure}
 
To illustrate why $\succ_u$ over $X$ need not be transitive, consider the following example: there are three applications with $X_1=\{x_1,x_2\}$, $X_2=\{x_2, x_3\}$ and $X_3=\{x_3, x_1\}$. If preferences are $U_1(x_1)> U_1(x_2)$, $U_2(x_2)>U_2(x_3)$, and $U_3(x_3)>U_3(x_1)$, then $\succ_u$ generates a cycle: $x_1 \succ_u x_2 \succ_u x_3 \succ_u x_1$ (see Figure~\ref{fig:cycle-intransitivity}).
 
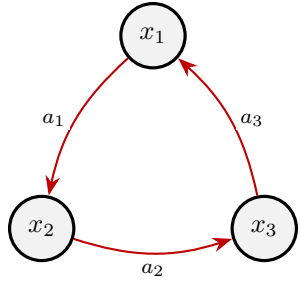
\begin{figure}[ht]
  \centering
  \begin{tikzpicture}
    \node[txn, multi] (x1) at (90:1.7) {$x_1$};
    \node[txn, multi] (x2) at (210:1.7) {$x_2$};
    \node[txn, multi] (x3) at (330:1.7) {$x_3$};
    \draw[cyc] (x1) to[bend right=18] node[lab, left=2pt] {$a_1$} (x2);
    \draw[cyc] (x2) to[bend right=18] node[lab, below=2pt] {$a_2$} (x3);
    \draw[cyc] (x3) to[bend right=18] node[lab, right=2pt] {$a_3$} (x1);
  \end{tikzpicture}
  \caption{Intransitivity example: three applications each with two transactions, whose unanimous preferences generate the cycle $x_1 \succ_u x_2 \succ_u x_3 \succ_u x_1$. Each comparison is supported by a single opinionated application (edge labels); every transaction interacts with two opinionated applications.}
  \label{fig:cycle-intransitivity}
\end{figure}

How to handle cycles is the core of the algorithm presented below. It is therefore useful to establish a structural fact about cycles that I use repeatedly.

\begin{lemma}\label{lem:two-multi-app}
Every cycle of $\succ_u$ contains at least two multi-opinion transactions.
\end{lemma}
For an informal proof,\footnote{An informal proof suffices because the derivations below use only a weaker statement: a cycle cannot exist without multi-opinion transactions.} consider a set containing single-opinion transactions and at most one multi-opinion transaction $m$. Partition the single-opinion transactions by the opinionated application they interact with. If the multi-opinion transaction $m$ exists, add it to every subset whose application it interacts with. Within each subset, all unanimity comparisons are governed by a single transitive preference, and there are no unanimity comparisons between different subsets. A cycle could therefore move from one subset to another only through $m$, which it visits at most once; the cycle would then lie within a single subset, contradicting transitivity.

\subsection{The algorithm}
 
Unanimity override determines the execution order $\succ_e$ in two steps:
\begin{enumerate}
    \item if the unanimity relation contains cycles, a fallback rule breaks them, leaving an acyclic repaired unanimity relation $\tilde\succ_u$;
    \item the repaired unanimity relation is extended into a total execution order by Kahn's algorithm, which uses $U_d$ to select among the available transactions.
\end{enumerate}
The first step determines which application-generated comparisons survive and is the substantive part of the mechanism. The second step completes the resulting acyclic relation. The main guarantees of the next section depend only on the fact that the completion extends the repaired relation, not on the particular procedure used; the procedure matters, however, for the transactions over which applications express no preferences (Section~\ref{sec:no-opinion}).
 
\subsubsection{Fallback rule in case of cycles}\label{subsec:fallback}
 
The possibility of cycles requires introducing a fallback rule aimed at recovering an acyclic ordering. Before presenting this rule, note that cycles can overlap, so the fallback does not operate on individual cycles but on each \emph{cyclic component} $Y\subseteq X$: a maximal set of transactions whose induced unanimity graph is strongly connected and contains a directed cycle (equivalently, a non-trivial strongly connected component of $\succ_u$).  There may be several cyclic components, and the fallback is applied to each separately.
 
The building block of the fallback is \emph{demotion}. The rule selects a set $D\subseteq Y$ to demote; the remaining transactions $K=Y\setminus D$ stay ordered by the surviving unanimity comparisons.
 
\begin{defn}[Demotion]
Fix a cyclic component $Y$ and a set $D\subseteq Y$, and let $K=Y\setminus D$. Demoting $D$ means: 
\begin{itemize}
    \item  removing every comparison in $\succ_u$ between two transactions in $Y$ whenever at least one of them lies in $D$;
    \item adding $  x\mathrel{\tilde\succ_u}d$ for every $x\in K$ and $d\in D$;
\item   adding $ d_i\mathrel{\tilde\succ_u}d_j$ whenever $d_i,d_j\in D$  and $U_d(d_i)>U_d(d_j)$.
\end{itemize}
Every comparison involving a transaction outside $Y$ is left unchanged.
\end{defn}
 
Three remarks are useful. First, demotion necessarily breaks some unanimity comparisons: it removes every unanimity comparison between the demoted transaction and any other element of $Y$. Second, demotion does not just relax constraints; it also introduces new ones that the execution order must respect, even though they do not arise from application unanimity. The repaired relation is therefore a mixture of surviving unanimity comparisons and protocol-imposed comparisons. Third, transactions in $D$ are ordered by the default order rather than by unanimity. Retaining their unanimity comparisons could recreate a cycle inside $D$, which is what the fallback must avoid.
 
Different demotion-selection rules are possible, and the main results do not depend on the specific rule as long as only multi-opinion transactions are demoted. Throughout, I adopt the following rule.
 
\medskip
\noindent\textbf{Fallback rule.} Starting from $D=\emptyset$, consider the relation induced by $Y\setminus D$. While that relation contains a cycle, add to $D$ the lowest-$U_d$ multi-opinion transaction that lies on a cycle in $Y\setminus D$. This set is nonempty by Lemma~\ref{lem:two-multi-app}. Ties are broken by a fixed deterministic rule, such as the transaction hash. Stop when the relation induced by $Y\setminus D$ is acyclic. A single demotion may not suffice: several transactions may be demoted before all cycles disappear.
\medskip
 
The fallback therefore uses only $U_d$ to select which multi-opinion transaction to demote. When $U_d$ is the priority fee, the rule has a simple economic interpretation: among the transactions eligible for demotion, it selects the one that paid the least, just as a fee market would. More sophisticated rules that use the topology of $Y$ to choose the demoted transaction are also possible, but I do not explore them here.
 
\begin{lemma}\label{lem:termination}
For every cyclic component $Y$, the fallback rule terminates after finitely many demotions. After the rule is applied to every cyclic component, the repaired relation $\langle X,\tilde\succ_u\rangle$ is acyclic.
\end{lemma}
The proof is in Appendix~\ref{app:proof-lem-termination}. The key observation is that, by Lemma~\ref{lem:two-multi-app}, a cycle cannot exist without at least two multi-opinion transactions. Therefore, whenever a residual cycle remains, the fallback has an eligible transaction to demote. Since the block is finite, the procedure eventually removes enough multi-opinion transactions to eliminate all
cycles.
 
\begin{lemma}[preservation]\label{lem:preservation}
If a transaction $x$ is not demoted by the fallback, then:
\[
  x\succ_u y \quad\Longrightarrow\quad x\mathrel{\tilde\succ_u}y.
\]
\end{lemma}
The proof is in Appendix~\ref{app:proof-lem-preservation}. Hence, if a transaction $x$ ranks higher than another transaction $y$ in the unanimity ranking and $x$ is not demoted, this ranking is preserved by the repaired unanimity relation. If $y$ is not demoted, the result is immediate. If $y$ is demoted, the result follows from the fact that every non-demoted transaction ranks higher than every demoted transaction.

The next proposition uses the above lemma to show how the fallback rule determines the ranking between any single-opinion transaction and the other transactions interacting with the same application.
\begin{proposition}[single-opinion transaction anchor]
\label{prop:single-opinion-anchor}
Fix an opinionated application $a_t$ and a transaction $x\in X_t$ such that $T^{\mathrm{o}}(x)=\{t\}$. For every $y\in X_t\setminus\{x\}$:
\begin{enumerate}[label=(\roman*)]
  \item if $U_t(x)>U_t(y)$, then $x\succ_u y$ and $x\mathrel{\tilde\succ_u}y$;
  \item if $U_t(y)>U_t(x)$, then $y\succ_u x$, and if $y$ is not demoted, then $y\mathrel{\tilde\succ_u}x$.
\end{enumerate}
Consequently, $\tilde\succ_u$ preserves every strict preference of $a_t$ that ranks the single-opinion transaction $x$ above another transaction. In particular, every strict comparison between two single-opinion transactions whose sole opinionated application is $a_t$ is preserved.
\end{proposition}
The proof is in Appendix~\ref{app:proof-prop-single-opinion-anchor}. For intuition, suppose transactions $x$ and $y$ interact with a given application. If one of the two is single-opinion, then this application is the only one that can express a preference over the ranking of $x$ and $y$, so the unanimity ranking coincides with that application's preferences. The proof then follows from Lemma~\ref{lem:preservation} and the fact that $x$ is single-opinion and hence is never demoted.

\subsubsection{Extending the repaired unanimity partial order}
 
It remains to extend the repaired unanimity relation into a total execution order. Here I use Kahn's topological-sort algorithm \citep{kahn1962topological}. Intuitively, the procedure builds the block from the top. Call a transaction \emph{available} if it has not yet been scheduled and no unscheduled transaction is ranked above it in the repaired relation. At each step, the procedure schedules the available transaction with the highest default-order value; scheduling a transaction may make new transactions available, and these compete immediately, by the default order, with those already waiting. The procedure stops when every transaction is scheduled, and the execution order $\succ_e$ ranks transactions in the order in which they were scheduled. Note that an isolated transaction --- one with no comparison in the repaired relation --- is available from the first step onward. Because a transaction is scheduled only after every transaction ranked above it, the resulting execution order extends the repaired relation: $\tilde\succ_u\subseteq\succ_e$ \citep{kahn1962topological}.

To illustrate the algorithm, let $X=\{x_1,x_2,x_3,x_4\}$ and suppose that the repaired graph contains
\[
  x_1\mathrel{\tilde\succ_u}x_2,
  \qquad
  x_1\mathrel{\tilde\succ_u}x_3,
\]
while $x_4$ is isolated. The first three transactions form one connected component, with  $x_2$ and $x_3$ incomparable. Let
\[
  U_d(x_1)>U_d(x_2)>U_d(x_4)>U_d(x_3).
\]
 
At the first step, $x_1$ and $x_4$ are available, and $x_1$ has the higher default value, so it is scheduled first. Scheduling $x_1$ makes $x_2$ and $x_3$ available, and $x_2$ now outbids the other available transactions. The schedule continues with $x_2$, then $x_4$, then $x_3$:
\[
  x_1\succ_e x_2\succ_e x_4\succ_e x_3.
\]
Note that $x_2$ is available only from the second step but executes before $x_4$. Had $U_d(x_2)>U_d(x_4)>U_d(x_1)$ held instead, $x_4$ would have been scheduled first, and $x_2$ --- despite paying more than $x_4$ --- would have executed after it, held back by its lower-fee predecessor $x_1$. Two transactions that the repaired relation leaves incomparable are therefore not, in general, ranked by the default order. The next lemma identifies when they are.

\begin{lemma}[default ranking]\label{lem:default-ranking}
Suppose no transaction is ranked above $x$ in the repaired relation $\tilde\succ_u$. Then $y\succ_e x$ only if $U_d(y)>U_d(x)$. In particular, two transactions with no predecessors in the repaired relation are always ranked by the default order.
\end{lemma}
The proof is immediate: $x$ is available from the first step of the algorithm and remains available until it is scheduled, so any transaction scheduled before $x$ was selected while $x$ was available, and therefore has a higher default value. If $y$ also has no predecessor, the same argument applies in both directions, so the pair executes in default order.
 
 For future reference, note that the main guarantees of Section~\ref{sec:threat-model} --- on which unanimous preferences are respected --- do not depend on how the repaired unanimity relation is extended. Section~\ref{sec:no-opinion}, however, discusses whether an attacker can front-run an isolated transaction, one over which applications express no preferences. The results derived there depend on the Kahn algorithm.
 
\subsection{Examples}
 Suppose there are two applications and six transactions, with
\[
  X_1=\{x_1,x_2,x_5,x_6\} \qquad X_2=\{x_3, x_4, x_5, x_6\}.
\]
So $x_5$ and $x_6$ interact with both applications, while each of the others interacts with only one.
 
\medskip
\noindent\emph{Example 1:} assume $U_1(x_1)>U_1(x_5)>U_1(x_6)>U_1(x_2)$ and $U_2(x_3)>U_2(x_6)>U_2(x_5)>U_2(x_4)$. Unanimity forces $x_1$ and $x_3$ before $x_5$ and $x_6$, and both before $x_2$ and $x_4$. Figure~\ref{fig:example1} shows the layered unanimity graph.

The algorithm schedules $x_1$ and $x_3$ first, then $x_5$ and $x_6$, then $x_2$ and $x_4$; within each pair, the default order decides. Note that no application ranks the pairs $(x_1, x_3)$ and $(x_2, x_4)$, so for them the default order acts as a tie-breaking rule. For the pair $(x_5, x_6)$, however, the applications have opposing preferences, so how the default order ranks these two transactions matters to them. 
 
\begin{figure}[ht]
  \centering
  \begin{tikzpicture}
    \node[txn] (x1) at (-1.8,3.2) {$x_1$};
    \node[txn] (x3) at (1.8,3.2) {$x_3$};
    \node[txn, multi] (x5) at (-1.8,1.6) {$x_5$};
    \node[txn, multi] (x6) at (1.8,1.6) {$x_6$};
    \node[txn] (x2) at (-1.8,0) {$x_2$};
    \node[txn] (x4) at (1.8,0) {$x_4$};
    \draw[aone] (x1) -- (x5);
    \draw[aone] (x1) -- (x6);
    \draw[atwo] (x3) -- (x5);
    \draw[atwo] (x3) -- (x6);
    \draw[aone] (x5) -- (x2);
    \draw[aone] (x6) -- (x2);
    \draw[atwo] (x5) -- (x4);
    \draw[atwo] (x6) -- (x4);
    \draw[unranked] (x5) -- node[lab] {disagree} (x6);
  \end{tikzpicture}
  \caption{Example 1: the unanimity graph (transitively implied comparisons omitted) is layered: $\{x_1, x_3\}$ are executed before $\{x_5, x_6\}$, which are executed before $\{x_2, x_4\}$. Within each layer the pair is unranked and decided by the default; on $(x_5, x_6)$ the applications genuinely disagree, so either default choice violates one application's preference.}
  \label{fig:example1}
\end{figure}
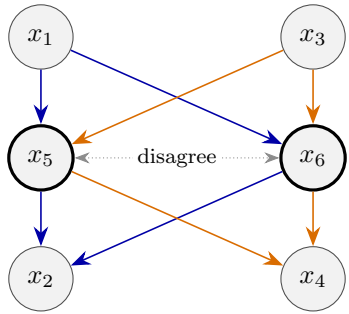

\medskip
\noindent\emph{Example 2:} preferences are $U_1(x_1)>U_1(x_5)>U_1(x_2)>U_1(x_6)$ and $U_2(x_6)>U_2(x_3)>U_2(x_5)>U_2(x_4)$, so the applications disagree over the ordering of $x_5$ and $x_6$. The unanimity relation contains the cycle $x_5 \succ_u x_2 \succ_u x_6 \succ_u x_3 \succ_u x_5$; the cyclic component is $Y = \{x_5, x_2, x_6, x_3\}$, whose multi-opinion transactions are $x_5$ and $x_6$ --- confirming Lemma~\ref{lem:two-multi-app}. Figure~\ref{fig:example2}(a) shows the cyclic component.
 
\emph{Fallback}: suppose $U_d(x_6) < U_d(x_5)$. The first round selects $x_6$; the relation restricted to $\{x_5, x_2, x_3\}$ is cycle-free, so the selection stops with $D = \{x_6\}$ and $K = \{x_5, x_2, x_3\}$. Demoting $x_6$ removes $x_6 \succ_u x_3$ and places $x_6$ below $x_5$, $x_2$, and $x_3$. As a by-product, the imposed comparison $x_5 \mathrel{\tilde \succ_u} x_6$ settles the pair on which the applications disagreed, in application 1's favor. Had $U_d$ instead ranked $x_5$ lowest, $x_5$ would have been demoted and the disputed pair settled in application 2's favor. Figure~\ref{fig:example2}(b) shows the repaired relation. The role of the Kahn procedure is simply to rank $x_1$ and $x_3$ according to the default order.
\begin{figure}[ht]
  \centering
  \begin{tikzpicture}
    \node[txn] (x1) at (0,3.4) {$x_1$};
    \node[txn, multi] (x5) at (-1.2,1.9) {$x_5$};
    \node[txn] (x3) at (1.2,1.9) {$x_3$};
    \node[txn] (x2) at (-1.2,0.3) {$x_2$};
    \node[txn, multi] (x6) at (1.2,0.3) {$x_6$};
    \node[txn] (x4) at (0,-1.2) {$x_4$};
    \draw[aone] (x1) -- (x5);
    \draw[cyc] (x5) -- (x2);
    \draw[cyc] (x2) -- (x6);
    \draw[cyc] (x6) -- (x3);
    \draw[cyc] (x3) -- (x5);
    \draw[atwo] (x6) -- (x4);
    \begin{scope}[on background layer]
      \node[grp, fit=(x5)(x2)(x6)(x3),
            label={[font=\scriptsize, gray]left:{$Y$}}] {};
    \end{scope}
    \node[font=\scriptsize] at (0,-2.3) {(a) unanimity relation $\succ_u$};
  \end{tikzpicture}
  \hspace{1.6cm}
  \begin{tikzpicture}
    \node[txn] (x1) at (-0.9,3.4) {$x_1$};
    \node[txn] (x3) at (0.9,3.4) {$x_3$};
    \node[txn, multi] (x5) at (0,2.2) {$x_5$};
    \node[txn] (x2) at (0,1.0) {$x_2$};
    \node[txn, multi] (x6) at (0,-0.2) {$x_6$};
    \node[txn] (x4) at (0,-1.4) {$x_4$};
    \draw[aone] (x1) -- (x5);
    \draw[atwo] (x3) -- (x5);
    \draw[aone] (x5) -- (x2);
    \draw[imposed] (x2) -- (x6);
    \draw[imposed] (x5) to[bend left=42] (x6);
    \draw[imposed] (x3) to[bend left=55] (x6);
    \draw[atwo] (x6) -- (x4);
    \node[font=\scriptsize] at (0,-2.3) {(b) repaired relation $\tilde\succ_u$};
  \end{tikzpicture}
 \if0 \hspace{1.2cm}
  \begin{tikzpicture}
    \node[txn] (x1) at (-0.7,3.4) {$x_1$};
    \node[txn] (x3) at (0.7,3.4) {$x_3$};
    \node[txn, multi] (x5) at (0,2.1) {$x_5$};
    \node[txn] (x2) at (0,0.8) {$x_2$};
    \node[txn, multi] (x6) at (0,-0.5) {$x_6$};
    \node[txn] (x4) at (0,-1.8) {$x_4$};
    \draw[uedge] (x1) -- (x5);
    \draw[uedge] (x3) -- (x5);
    \draw[uedge] (x5) -- (x2);
    \draw[uedge] (x2) -- (x6);
    \draw[uedge] (x6) -- (x4);
    \draw[unranked] (x1) -- node[lab] {default} (x3);
    \node[font=\scriptsize] at (0,-2.9) {(c) execution order $\succ_e$};
  \end{tikzpicture}\fi
  \caption{Example 2. (a) The unanimity relation contains the cycle $x_5 \succ_u x_2 \succ_u x_6 \succ_u x_3 \succ_u x_5$ (red); the cyclic component is $Y = \{x_5, x_2, x_6, x_3\}$ and its multi-opinion transactions (thick circles) are $x_5$ and $x_6$ (transitively implied comparisons are omitted). (b) With $U_d(x_6) < U_d(x_5)$ the fallback demotes $D = \{x_6\}$: every comparison between $x_6$ and the rest of $Y$ is removed, and the dashed comparisons placing $x_6$ below $K = \{x_5, x_2, x_3\}$ are imposed.
  }
  \label{fig:example2}
\end{figure}
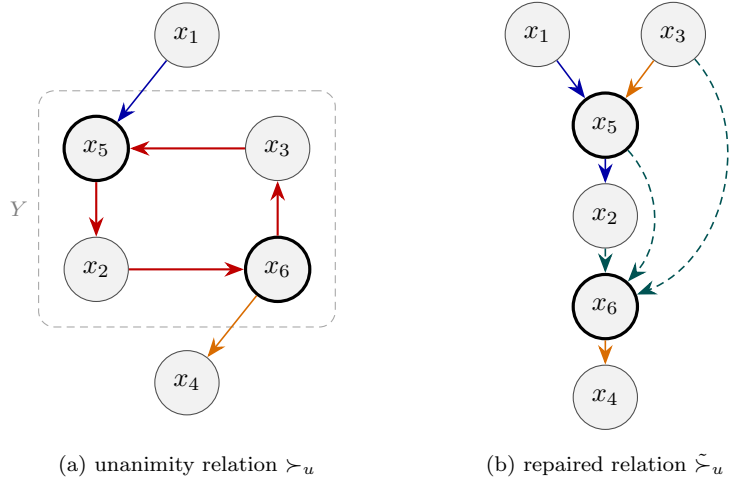
\section{Threat model and guarantees}\label{sec:threat-model}

I now specify the threat model. The attacker can deploy one or more applications, choosing each application's preferences and interaction domain; craft transactions, choosing which applications they validly interact with and how each application ranks them; and, in the worst-case environment, choose the default order.\footnote{
Later I specialize the default order to priority-fee ordering, in which case the attacker can change the default-order position only of transactions whose fee it controls.} The attacker cannot modify a transaction submitted by another player, change any of its interactions, or remove it from the block. This is the strong form of censorship resistance assumed in the introduction. Nor can the attacker change the preferences or interaction domain of an application already present in the baseline; it may only deploy additional applications.

Not every interaction is available to the attacker: it cannot interact from
an address it does not control, and an application may accept some
interactions only from designated senders. Let $R\subseteq\mathbb X$ be the
set of transactions the attacker can craft. Say that a transaction $x'$
\emph{extends} $x$ if every application that $x$ interacts with also ranks
$x'$, and ranks it identically: $x\in\mathbb X_t$ implies $x'\in\mathbb X_t$
and $U_t(x')=U_t(x)$. An extension may interact with further applications;
every transaction is a (trivial) extension of itself. I impose two 
assumptions on $R$:
\begin{enumerate}[label=(R\arabic*)]
  \item \emph{(single-opinion decomposition)} every transaction in $R$ can be decomposed into equivalent single-opinion transactions, themselves elements of $R$. Formally, for every $z\in R$ and every $t\in
  T^{\mathrm o}(z)$, $R$ contains a transaction $z_t$ with $T^{\mathrm
  o}(z_t)=\{t\}$ and $U_t(z_t)=U_t(z)$.
  \item \emph{(bundling)} two transactions in $R$ can be bundled into a single transaction that is equivalent to the two, and this single transaction is also an element of $R$. Formally, for all $z,w\in R$ interacting with disjoint sets
  of applications, $R$ contains an extension of $z$ that every application
  interacting with $w$ ranks as it ranks $w$.
\end{enumerate}
Both assumptions reflect how transactions are built in practice: a
transaction is a bundle of interactions, and the attacker can compose any
interactions within its reach --- keeping one, dropping the others, or
combining two bundles into one. Finally, a newly deployed application's
domain contains only transactions that interact with it: the attacker cannot
place an honest transaction in the domain of an application it deploys.

\paragraph{Baseline profile and replacement.}
Fix a baseline profile: a block $X$, the applications and their preferences, the interaction sets, and the default order. Let $\succ_u$, $\tilde\succ_u$, and $\succ_e$ denote the corresponding unanimity relation, repaired unanimity relation, and execution order. The baseline block may already contain attacker-owned transactions: let $C\subseteq X\cap R$ be the set of baseline transactions the attacker controls.

The attacker cannot modify the baseline transactions of other players, but it
can modify  its own. An attacker strategy $\sigma$ may replace any
transaction in $C$ by a fresh extension of it in $R$, and may inject a set
$Z^\sigma\subseteq R$ of additional fresh transactions; fresh means distinct
from every other transaction in the attacked block. 
Write $\kappa^\sigma(x)$ --- abbreviated $\bar x$ when the strategy is clear
from context --- for the transaction that stands in for a baseline
transaction $x$ in the attacked block: its replacement if it is replaced,
and $x$ itself otherwise. In either case $\kappa^\sigma(x)$ is an extension
of $x$ --- a fresh one, or the trivial one. The attacked block and the attacker-controlled transactions in it are
\[
  X^\sigma=(X \setminus C)\cup C^\sigma,
  \qquad
  C^\sigma=\kappa^\sigma(C)\cup Z^\sigma.
\]
The attacker may additionally deploy applications and, in the worst-case environment, choose the default order; let $\succ_u^\sigma$, $\tilde\succ_u^\sigma$, and $\succ_e^\sigma$ denote the resulting relations. 

\begin{defn}[attack]
Fix a baseline profile and its controlled set $C$. A \emph{target pair} is a pair of distinct transactions $(x_i,x_j)$ in $X$ such that
\[
  x_i\succ_u x_j
  \qquad\text{and}\qquad
  x_i\mathrel{\tilde\succ_u}x_j.
\]
Thus the applications unanimously rank $x_i$ above $x_j$, and the baseline fallback preserves this comparison. A feasible strategy $\sigma$ is a \emph{successful attack} on $(x_i,x_j)$ if
\[
  \kappa^\sigma(x_j)\succ_e^\sigma x_i.
\]
\end{defn}

The success condition is asymmetric: it tracks the stand-in $\kappa^\sigma(x_j)$ but the original $x_i$, which embeds the idea that $x_i$ is not controlled by the attacker.  If $x_j\in C$, a successful attack is a \emph{front-running attack}: the attacker advances the extension of its own lower-ranked baseline transaction. If $x_j\notin C$, then $ \kappa^\sigma(x_j)=x_j$, and a successful attack is \emph{griefing}.

\subsection{How attacks can be carried out}
The definition admits two distinct ways of defeating a protected unanimity
comparison: adding unanimity comparisons to create a cycle and induce
demotion, or removing unanimity comparisons by extending $x_j$ so that
applications conflict over the attacked pair, and then relying on the
completion step. I discuss each channel in turn.

\paragraph{First channel: cycle creation and demotion.}
Suppose that the attacked profile still contains the unanimous comparison
\[
  x_i\succ_u^\sigma\bar x_j.
\]
Kahn's algorithm extends the repaired relation:
\[
  \tilde\succ_u^\sigma\subseteq\succ_e^\sigma.
\]
Hence the completion step cannot reverse this comparison if it survives the fallback. A successful attack through this channel must place $x_i$ and $\bar x_j$ in the same cyclic component and cause the fallback to remove the comparison by demoting $x_i$. The attack therefore proceeds in two steps:
\begin{enumerate}
  \item \emph{cycle creation}: embed the target pair in a cyclic component by injecting transactions positioned within the rankings of honest applications and, possibly, deploying an application to close the cycle;
  \item \emph{demotion targeting}: make the fallback demote $x_i$ and, if $\bar x_j$ is also demoted, make the default order rank $\bar x_j$ above $x_i$.
\end{enumerate}

As an illustrative example, consider two honest applications $a$ and $b$ that both rank $x_i$ above $x_j$, so $x_i\succ_u x_j$. The attacker leaves the target transactions unchanged, deploys one application $M$, and injects two transactions: $y_1$, interacting with $a$ and $M$, whose interaction with $a$ satisfies $U_a(x_j)>U_a(y_1)$; and $y_2$, interacting with $b$ and $M$, whose interaction with $b$ satisfies $U_b(y_2)>U_b(x_i)$. Since $M$ is the only application ranking $(y_1,y_2)$, the attacker chooses its interactions with $M$ such that $y_1\succ_u^\sigma y_2$. The resulting cycle is
\[
  x_i
  \succ_u^\sigma x_j
  \succ_u^\sigma y_1
  \succ_u^\sigma y_2
  \succ_u^\sigma x_i.
\]
Figure~\ref{fig:attack} shows the cycle and the demotion of $x_i$ when the default order ranks it lowest. 

To conclude, note that this example requires three things. The attacker must
be able to craft a transaction that some application ranks strictly below
the pair's lower transaction --- here, $y_1$ below $x_j$ --- and a
transaction that another application ranks strictly above the target
transaction --- here, $y_2$ above $x_i$. And the target transaction must be
multi-opinion, and therefore eligible for demotion. The next subsection
shows that these conditions are at the core of the protocol's guarantees:
when the attacker cannot reach above the target transaction or below the
pair's lower transaction, or when the target transaction is single-opinion,
this attack is impossible. 

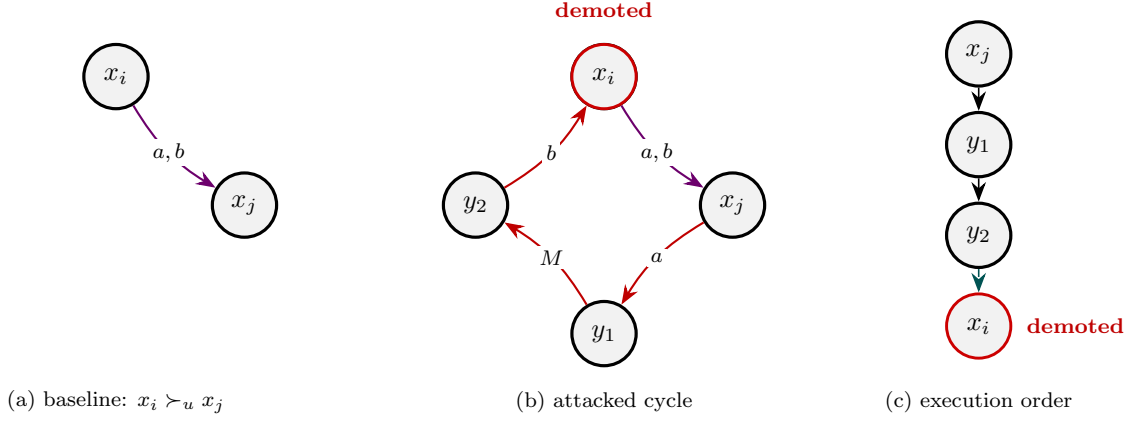
\begin{figure}[ht]
  \centering
  \begin{tikzpicture}
    \useasboundingbox (-2.5,-2.9) rectangle (2.5,2.7);
    \node[txn, multi] (xi) at (90:1.7) {$x_i$};
    \node[txn, multi] (xj) at (0:1.7) {$x_j$};
    \draw[aboth] (xi) to[bend right=14] node[lab] {$a,b$} (xj);
    \node[font=\scriptsize] at (0,-2.6) {(a) baseline: $x_i\succ_u x_j$};
  \end{tikzpicture}
  \hspace{1.2cm}
  \begin{tikzpicture}
    \useasboundingbox (-2.5,-2.9) rectangle (2.5,2.7);
    \node[txn, multi] (xi) at (90:1.7) {$x_i$};
    \node[txn, multi] (xj) at (0:1.7) {$x_j$};
    \draw[aboth] (xi) to[bend right=14] node[lab] {$a,b$} (xj);
    \node[txn, multi] (y1) at (-90:1.7) {$y_1$};
    \node[txn, multi] (y2) at (180:1.7) {$y_2$};
    \draw[cyc] (xj) to[bend right=14] node[lab] {$a$} (y1);
    \draw[cyc] (y1) to[bend right=14] node[lab] {$M$} (y2);
    \draw[cyc] (y2) to[bend right=14] node[lab] {$b$} (xi);
    \node[txn, multi, draw=red!80!black, line width=1.1pt] at (90:1.7) {$x_i$};
    \node[font=\scriptsize\bfseries, red!75!black, anchor=south] at (90:2.35) {demoted};
    \node[font=\scriptsize] at (0,-2.6) {(b) attacked cycle};
  \end{tikzpicture}
  \hspace{1.0cm}
  \begin{tikzpicture}
    \useasboundingbox (-1.2,-2.9) rectangle (2.2,2.7);
    \node[txn, multi] (xj) at (0,2.0) {$x_j$};
    \node[txn, multi] (y1) at (0,0.8) {$y_1$};
    \node[txn, multi] (y2) at (0,-0.4) {$y_2$};
    \node[txn, multi, draw=red!80!black, line width=1.1pt] (xi) at (0,-1.6) {$x_i$};
    \draw[uedge] (xj) -- (y1);
    \draw[uedge] (y1) -- (y2);
    \draw[imposed] (y2) -- (xi);
    \node[font=\scriptsize\bfseries, red!75!black, anchor=west] at (0.5,-1.6) {demoted};
    \node[font=\scriptsize] at (0,-2.6) {(c) execution order};
  \end{tikzpicture}
  \caption{The cycle-and-demotion attack. The attacker preserves the target transactions, creates a cycle around the baseline edge, and makes the fallback demote $x_i$. Kahn then extends a repaired relation that places $x_j$ above $x_i$.}
  \label{fig:attack}
\end{figure}

\paragraph{Second channel: replacement and Kahn completion.}
Consider again a target pair $(x_i,x_j)$. Now assume that $x_j\in C$ and the attacker replaces it, so that
$\bar x_j=\kappa^\sigma(x_j)$ is an extension of $x_j$. By definition of extension, every baseline application shared by $x_i$ and $x_j$
retains its comparison, including at least one strict preference for $x_i$. However, the replacement may succeed in destroying the baseline comparison by adding
an opposing interaction with another application that $x_i$ interacts with. 

For example, suppose that the baseline block is $X=\{x,y\}$, with $C=\{y\}$. There are two honest opinionated applications, $a$ and $b$. Transaction $x$ interacts with both applications, while $y$ interacts only with $a$, and
\[
  U_a(x)>U_a(y).
\]
The only application shared by the baseline pair is $a$, so
\[
  x\succ_u y.
\]

Suppose the attacker's reach contains a transaction $w$ with
$U_b(w)>U_b(x)$. By bundling, $R$ contains an extension of $y$ that $b$
ranks as it ranks $w$; let the attacker replace $y$ with it, so that this
extension is $\bar y=\kappa^\sigma(y)$, with
\[
  U_a(\bar y)=U_a(y)<U_a(x)
\]
and
\[
  U_b(\bar y)>U_b(x).
\]
The applications now disagree over $(x,\bar y)$, so there is no unanimity edge in either direction. Since the attacked block is $\{x,\bar y\}$, both transactions are isolated. If
\[
  U_d(\bar y)>U_d(x),
\]
Kahn returns
\[
  \bar y\succ_e^\sigma x.
\]
This is a successful attack on the protected baseline pair $(x,y)$.

Note that the same three ingredients discussed in the context of the other attack appear also here. The attack needs the target transaction to be
multi-opinion; it needs a transaction ranked below it at one
application (here $\bar y$); and it needs a reachable
transaction ranked above it at another application (here the
interaction with $b$ that the extension adds). Again, when these conditions are not satisfied, this attack is also infeasible.

\subsection{Unanimity override's guarantees}
The two channels fail under identifiable conditions, and those conditions
are the protocol's guarantees. First, if no transaction the attacker can
craft ranks strictly above the target transaction at any of its
applications, neither channel can supply the ``above'' transaction: the
target transaction is \emph{gated}, and the attack always fails. Second, a
symmetric guarantee holds at the bottom of the order: a transaction that is
ranked at the bottom by all applications cannot be attacked. Third, if the
target transaction is single-opinion, it can neither be demoted nor be made
to conflict across applications: its unanimity comparisons survive intact.
The next propositions state these guarantees formally.
Section~\ref{sec:impossibility} shows that these guarantees cannot be
improved upon.

\begin{proposition}[gated region]\label{thm:protected-region}
Let $X^*\subseteq X$ be the largest set such that, under every feasible
attacker strategy $\sigma$:
\begin{itemize}
  \item \emph{(unreachable)} no transaction in $X^\sigma\setminus X^*$ is
  unanimously ranked above a member of $X^*$; and
  \item \emph{(acyclic)} the restriction of $\succ_u^\sigma$ to $X^*$ is
  acyclic.
\end{itemize}
Then, under every feasible strategy $\sigma$, for every $x^*\in X^*$ and
every $z\in X$ with $x^*\succ_u z$, we have
$x^*\succ_e^\sigma\kappa^\sigma(z)$. Consequently, no target pair whose
higher-ranked transaction is gated can be attacked, through either channel,
whatever the default order.
\end{proposition}
The proof is in Appendix~\ref{app:proof-prop-top-transaction}. The set $X^*$ is well-defined because both defining conditions are preserved under arbitrary unions.\footnote{If a transaction outside a union $\bigcup_\alpha S_\alpha$ were unanimously ranked above a member, it would be ranked above a member of some $S_\alpha$, violating that set's unreachability. For acyclicity, choose a member of any purported cycle in some $S_\alpha$ and follow the cycle backward. Unreachability of $S_\alpha$ forces every predecessor, and hence the entire cycle, to lie in $S_\alpha$, contradicting its acyclicity.} 

The proposition covers two motivating examples from the introduction. A
PropAMM accepts oracle updates only from a designated address and ranks the
update above every swap: no transaction the attacker can craft ranks above
it, so it is gated. The same argument applies to a cancellation spanning
several markets: if every market ranks cancellations first and accepts them
only from the order's owner, no reachable transaction outranks the
cancellation at any market, and it is gated.

An identical guarantee holds symmetrically for bottom-gated transactions, that is, transactions the applications unanimously rank last among all possible transactions.
\begin{proposition}[bottom-gated region]\label{prop:bottom-gated}
Let $X_*\subseteq X$ be the largest set such that, under every feasible
attacker strategy $\sigma$: (i) no member of $X_*$ is unanimously ranked above
any transaction in $X^\sigma\setminus X_*$, and (ii) the restriction of
$\succ_u^\sigma$ to $X_*$ is acyclic. Then, under every feasible strategy
$\sigma$, for every $x_*\in X_*$ and every $z\in X$ with
$z\succ_u x_*$, we have $\kappa^\sigma(z)\succ_e^\sigma x_*$.
\end{proposition}
The proof is in Appendix~\ref{app:proof-prop-bottom-gated}. A settlement
transaction that every application ranks last is the natural example: under
the validity convention, no unanimity comparison points out of it, so it
executes after everything the applications rank above it.

The third guarantee requires no gating at all: a single-opinion transaction
executes according to the preferences of its application, whatever the
attacker does.

\begin{proposition}[single-opinion guarantee]\label{thm:single-opinion}
Let $x_i\in X$ satisfy $T^{\mathrm o}(x_i)=\{t\}$. Under every feasible strategy $\sigma$, for every $z\in X^\sigma$:
\[
  U_t(x_i)>U_t(z)
  \quad\Longrightarrow\quad
  x_i\succ_e^\sigma z.
\]
In particular, all single-opinion transactions are executed according to the preferences of the application each interacts with.
\end{proposition}
The proof is in Appendix~\ref{app:proof-lem-same-app-invariance}. The
intuition was already discussed when presenting the two attack channels: a
single-opinion transaction is never demoted, and it has no second
opinionated application at which the attacker's replacement can add an
opposing interaction.

Note that the above proposition does not settle which of an
application's transactions runs first. By Proposition~\ref{prop:single-opinion-anchor}, a multi-opinion transaction that is not demoted preserves its priority over each single-opinion anchor. Therefore, for every application, its first transaction is either its most-preferred single-opinion transaction or an even more preferred multi-opinion transaction. Going back to the example of an on-chain auction, the implication is that if the highest bid is embedded in a single-opinion transaction, unanimity override guarantees that this bid reaches the application before all other bids. If the highest bid is embedded in a multi-opinion transaction, the guarantee may fail. The sender's incentives then dictate that only single-opinion transactions be used, and the auction operates as intended.

It follows that unanimity override matches, and even strengthens, the
guarantee a rollup provides. A rollup gives an application full control over
the ordering of the transactions that interact with it, but only by
isolating it. By Proposition~\ref{thm:single-opinion}, unanimity override preserves every strict
comparison in which the higher-ranked transaction is single-opinion, even if
the lower-ranked transaction also interacts with other opinionated
applications. In particular, two transactions whose sole opinionated
application is $a_t$ are executed in the strict order specified by $a_t$.
The guarantee
does not require the transactions to interact only with $a_t$: they may also
touch non-opinionated applications, so composability is preserved.

\subsection{Attacks with and without unanimity override}
\label{subsec:attacks-with-without-override}

I now compare the attacks available with unanimity override to those available when transactions execute directly in the default order. The expanded attack definition makes the comparison depend both on cycle creation and on whether the lower-ranked baseline transaction is controlled and can be replaced.

\paragraph{Attacker-controlled default ordering.}
Without unanimity override, an attacker that controls $U_d$ can reverse
every target pair. With
unanimity override, a target pair cannot be reversed when its higher-ranked
transaction is single-opinion, by Proposition~\ref{thm:single-opinion}, or
when it is gated, by Proposition~\ref{thm:protected-region}; and every
comparison placing a transaction above a bottom-gated one survives, by
Proposition~\ref{prop:bottom-gated}. Hence, unanimity override strictly
reduces an attacker's ability to manipulate the execution order of
transactions.

\paragraph{Priority-fee ordering.}
When $U_d$ is the priority-fee order, the replacement channel prices
front-running as the benchmark does: in the two-transaction example of the
second channel, both transactions are isolated after the replacement, so by
Lemma~\ref{lem:default-ranking} the attack succeeds only if the extension
outbids the target transaction --- exactly the price of front-running under
priority-fee ordering. The cycle channel is where the comparison becomes
ambiguous: there is no unconditional cost ranking between unanimity override
and the benchmark, as the following two examples illustrate.

\medskip
\begin{exmpl}[Costly front-running]
Interpret $x_j$ in Figure~\ref{fig:attack} as an attacker-controlled transaction that the strategy leaves unchanged. Every transaction on the cycle
\[
  x_i\succ_u^\sigma x_j
  \succ_u^\sigma y_1
  \succ_u^\sigma y_2
  \succ_u^\sigma x_i
\]
is multi-opinion. For the fallback to select $x_i$ for demotion, $x_i$ must have the lowest priority among these transactions. In particular,
\[
  U_d(x_j)>U_d(x_i),
\]
so the front-running transaction must cross the same priority threshold as under pure priority-fee ordering. Moreover, the attacker transactions $y_1$ and $y_2$ must also outrank $x_i$. In this example, unanimity override adds transactions and constraints to the direct outbid required by the benchmark.
\end{exmpl}

The conclusion is not general, as the next example shows.

\begin{exmpl}[Cheap front-running]
\label{ex:low-fee-single-opinion-bridge}
Let $j\in C$ be an attacker-owned transaction already present in the baseline block, let $b\in R$ be an additional transaction the attacker can inject, and let $x,p,r$ be honest transactions. Consider three honest applications $A$, $B$, and $C$ with the following relevant rankings:
\[
\begin{aligned}
  x&\succ_A j\succ_A p,\\
  p&\succ_B b\succ_B r,\\
  r&\succ_C x\succ_C p.
\end{aligned}
\]
The transaction $b$ is absent from the baseline. Before $b$ is injected, the unanimity graph contains
\[
  r\succ_u x\succ_u j\succ_u p,
\]
together with $x\succ_u p$. Applications $B$ and $C$ disagree over $(p,r)$, so there is no edge between them. The baseline graph is acyclic and, in particular,
\[
  x\succ_e j.
\]

Now let the attacker leave $j$ unchanged and inject $b$. The unanimity graph acquires the cycle
\[
  r\succ_u^\sigma x
  \succ_u^\sigma j
  \succ_u^\sigma p
  \succ_u^\sigma b
  \succ_u^\sigma r.
\]
The transactions $x,p,r$ are multi-opinion, whereas $j$ and $b$ are single-opinion. Suppose the priority-fee order satisfies
\[
  U_d(j),U_d(b)<U_d(x)<U_d(p),U_d(r).
\]
Among the multi-opinion transactions on the cycle, $x$ has the lowest priority, so the fallback demotes $x$. The repaired relation places every non-demoted transaction in the component, including $j$, above $x$. Hence
\[
  j\succ_e^\sigma x,
\]
even though neither attacker-controlled transaction outranks $x$ by fee.
\end{exmpl}

Hence, manipulating the order of transactions may be cheaper or more
expensive under unanimity override than under priority-fee ordering. In particular, it is
cheaper when honest multi-opinion transactions with high fees can be
recruited into a cycle. Another important difference is griefing: under
priority-fee ordering an attacker cannot change the priority fee (and hence
the order) of transactions it does not control, but under unanimity override
the cycle channel makes griefing possible --- in the costly front-running
example, the attack goes through unchanged even if $x_j$ is not controlled
by the attacker. Finally, these considerations apply only to non-gated
multi-opinion transactions; gated, bottom-gated, and single-opinion
transactions are protected by the guarantees above.

\section{The limits of ordering guarantees}\label{sec:impossibility}

The guarantees of the previous section leave out one class: the
multi-opinion transactions that are gated neither above nor below. I now
argue that no ordering rule can guarantee to respect the applications'
preferences over those transactions, at least from the \emph{ex ante} ---
before the block is known --- viewpoint. The reason is that applications
may disagree over the ordering of multi-opinion transactions that are not
gated, and hence every execution order violates the preferences of some
application.

To establish this fact, I extend the richness assumptions (R1) and (R2) to
describe the transactions any sender could send, not only the attacker: a
transaction is a bundle of interactions, and any sender can decompose or
combine the bundles within its reach. I extend the notion of gating in the
same way. Call a possible transaction $x$ \emph{gated above at} $a_t$ if
no possible transaction ranks strictly above $x$ in the preferences of
$a_t$, and \emph{gated below at} $a_t$ if no possible transaction ranks
strictly below $x$; the transaction is \emph{gated above} if it is gated
above at every opinionated application it interacts with, and \emph{gated
below} symmetrically.

\begin{proposition}[impossibility]\label{prop:necessity}
Let $x$ be a possible transaction that is multi-opinion and gated neither
above nor below. Then there exist another possible transaction $y$, also
multi-opinion and gated neither above nor below, and two applications
$a_t$ and $a_{t'}$ such that
\[
  U_t(x)>U_t(y)
  \qquad\text{and}\qquad
  U_{t'}(y)>U_{t'}(x).
\]
\end{proposition}
The proof is in Appendix~\ref{app:proof-prop-necessity}. Note that no
single-opinion or gated transaction could play the role of $y$: a
single-opinion transaction shares at most one opinionated application with
$x$, so no second application can rank the pair in the opposite direction;
and a transaction gated above (below) admits no possible transaction
ranked strictly above (below) it, ruling out one of the two inequalities.
Conflicts of this kind therefore arise only within the class of transactions that the
guarantees leave out.

For intuition, because $x$ is multi-opinion and gated neither above nor
below, one application it interacts with ranks a possible transaction
strictly above $x$, and a different one ranks a possible transaction
strictly below $x$. By decomposition and bundling, a single possible
transaction $y$ can then bundle the two interactions. In any block
containing $x$ and $y$, the two applications disagree on their order: one
wants $y$ to execute first, and the other wants $x$. Both transactions are
multi-opinion and gated neither above nor below, so whichever way a rule
orders the pair, it violates one application's strict preference.

Finally, note that the result assumes nothing about the specific rule used
to aggregate the applications' preferences into an execution order: it
shows that no execution order --- however chosen --- satisfies all
applications' preferences.

\section{Protecting senders: transactions without application preferences}
\label{sec:no-opinion}
The guarantees derived so far protect preferences expressed by applications.
Senders, however, may also have preferences. When an application declares its
preferences and a sender submits a transaction to it, I treat the two as
aligned: respecting the application's preferences respects the sender's. The
converse need not hold. An application may express no preference --- a
traditional AMM, for example --- while the sender still has one, typically
not to be front-run. For these transactions the notion of attack must change:
there is no unanimous comparison to reverse, so the question is not whether a
protected pair can be flipped but whether another transaction can execute
first without outbidding the sender's.

Recall that a no-opinion transaction interacts with no opinionated application. Fix such a baseline transaction $x\in X$. Whenever $x$ appears in the attacked block, it appears unchanged: the attacker cannot make it interact with a newly deployed application or change the preferences of the non-opinionated applications it already interacts with. It is therefore isolated in the attacked profile.

\begin{proposition}[no-opinion transactions]\label{prop:no-opinion}
Let $U_d$ be the priority-fee order and fix a no-opinion transaction $x\in X$. Under every feasible attacker strategy $\sigma$, a transaction $y\in X^\sigma$ executes before $x$ only if
\[
  U_d(y)>U_d(x).
\]
\end{proposition}

The proof is immediate. If $x$ is not in the attacked block, no transaction executes before it and the claim is vacuous. Otherwise $x$ is kept unchanged, so it is isolated in $\succ_u^\sigma$, belongs to no cyclic component, and remains isolated in $\tilde\succ_u^\sigma$; it is then available from the first step of Kahn's algorithm and remains available until scheduled, and any transaction selected before it must pay a strictly higher priority fee, by Lemma~\ref{lem:default-ranking}.

Paying a strictly higher priority fee is therefore necessary to execute before a no-opinion transaction, exactly as under the priority-fee benchmark.

\section{Robustness to alternative aggregation rules}
\label{sec:robustness}
I now generalize the results derived above to any aggregation rule that is Paretian: if all applications agree on an order, the aggregation rule should follow it. Clearly, unanimity is Paretian, as are most standard aggregation rules --- majority voting, dictatorship, the Pareto-extension rule. I first argue that any Paretian aggregation rule is an extension of unanimity. And because the gating and single-opinion guarantees depend only on extending the repaired unanimity relation, not on how the extension is built, they continue to hold for every Paretian rule, provided cycles are broken lexicographically --- sacrificing non-unanimity comparisons first.
 
Formally, consider a generic \emph{aggregation rule} $\Phi$ with corresponding binary relation $\succ_\Phi$ on $X$. For a pair $(x_i,x_j)$, let
\[
  \mathcal A_{ij} \equiv \{t:x_i,x_j\in X_t\}
\]
be the applications that rank the pair, and split them by direction:
\[
  S^+_{ij}
  \equiv
  \{t\in\mathcal A_{ij}:U_t(x_i)>U_t(x_j)\},
  \qquad
  S^-_{ij}
  \equiv
  \{t\in\mathcal A_{ij}:U_t(x_j)>U_t(x_i)\}.
\]
The rest, $\mathcal A_{ij}\setminus(S^+_{ij}\cup S^-_{ij})$, are indifferent.
In this notation the unanimity relation is
\[
  x_i\succ_u x_j
  \quad\Longleftrightarrow\quad
  S^+_{ij}\neq\emptyset
  \ \text{and}\
  S^-_{ij}=\emptyset.
\]
 
\begin{defn}[Pareto principle]
    An aggregation rule $\Phi$ is \emph{Paretian} whenever 
    \[
  S^+_{ij}\neq\emptyset
  \ \text{and}\
  S^-_{ij}=\emptyset
  \quad\Longrightarrow\quad
  x_i\succ_\Phi x_j \quad\text{and}\quad x_j \not\succ_\Phi x_i.
\]
\end{defn}
In words: if at least one application ranking the pair strictly prefers $x_i$ to execute first, and none prefers the opposite, then the rule ranks $x_i$ above $x_j$. This is the strong Pareto principle, applied pair by pair to the applications that rank the pair.
 
\begin{proposition}[Paretian rules extend unanimity]
\label{prop:paretian-extension}
Every Paretian aggregation rule is an extension of unanimity:
\[
  \succ_u\ \subseteq\ \succ_\Phi.
\]
\end{proposition}
The proof is immediate from the definition of unanimity.

So every Paretian rule is unanimity plus some treatment of the pairs
unanimity leaves open. Unanimity override is therefore the minimal Paretian rule: it adds no comparison beyond unanimity itself. Because the unanimity relation may contain cycles, any Paretian rule may contain cycles too. To recover an acyclic order, it is necessary to introduce a fallback rule. 
 
To do so, call a comparison $x\succ_\Phi y$ a non-unanimity comparison if $x\not\succ_u y$. Thus non-unanimity comparisons are the comparisons added by $\Phi$ beyond unanimity. The repair of $\succ_\Phi$ proceeds in two stages. First, while the current relation contains a directed cycle with at least one non-unanimity comparison, delete one non-unanimity comparison on such a cycle. Ties are broken by a fixed deterministic rule. Stop when every remaining cycle, if any, consists only of unanimity comparisons. Second, apply the demotion fallback of Section~\ref{subsec:fallback} to the remaining cycles.

\begin{proposition}
\label{prop:repaired-phi-extends-repaired-unanimity}
Let $\Phi$ be Paretian, so that
\[
  \succ_u\subseteq \succ_\Phi .
\]
Repair $\succ_\Phi$ by first deleting non-unanimity comparisons on cycles until
no cycle contains a non-unanimity comparison, and then applying the demotion
fallback to the remaining cycles. Let the resulting relation be
$\tilde\succ_\Phi$. Then
\[
  \tilde\succ_u\subseteq \tilde\succ_\Phi .
\]
Moreover, $\tilde\succ_\Phi$ is acyclic.
\end{proposition}
 The proof is in Appendix~\ref{last-proof}.

 Note that the gating and single-opinion guarantees depend only on the fact that the execution order extends $\tilde\succ_u$; the precise extension plays no role. Hence they continue to hold for any Paretian rule, provided cycles are broken lexicographically: non-unanimity comparisons may be discarded to restore acyclicity before any unanimity comparison is sacrificed through demotion. The bottom-gated guarantee transfers as well: its proof uses only that a member of $X_*$ never lies in a cyclic component, and under the lexicographic repair the demotion stage operates on exactly the cyclic components of unanimity (Proposition~\ref{prop:repaired-phi-extends-repaired-unanimity}).
 
  The protection of Proposition~\ref{prop:no-opinion} extends as well, once isolation is read relative to the rule. Lemma~\ref{lem:default-ranking} concerns only the completion step, so it applies unchanged when the execution order extends $\tilde\succ_\Phi$: if $U_d$ is the priority fee, a transaction that is isolated in $\tilde\succ_\Phi$ under every feasible attacker strategy can be preceded only by transactions paying a strictly higher fee. What changes across rules is which transactions qualify. Under unanimity override, every no-opinion transaction does, because unanimity adds no comparison on pairs over which applications are silent. The Pareto principle does not constrain those pairs, so another Paretian rule may rank them, and a no-opinion transaction need not remain isolated under $\Phi$.

Finally, the impossibility result of Section~\ref{sec:impossibility} binds
every rule, Paretian or not: once two applications disagree over a pair, no
execution order satisfies both, whatever rule produced it
(Proposition~\ref{prop:necessity}). A richer Paretian rule changes which
preferences survive a conflict, not whether conflicts emerge.

\section{Interpreting the interaction domains}\label{sec:interaction-sets}
 
The analysis so far took the interaction domains $X_t$ as given. In practice they must be constructed, and there is more than one way to do it. The choice is not innocuous: it affects both the strength of the guarantees and whether a unique execution order exists at all.
 
\paragraph{Declared interactions.} The first option is to build $X_t$ from the applications each transaction \emph{could} interact with, declared before execution. On Solana, every transaction already lists the accounts it may access; on Monad, parallel execution tracks the same information to resolve conflicts. Under this reading $\mathbb{X}_t$ is the set of transactions whose declared access includes $a_t$, and $X_t = X \cap \mathbb{X}_t$ is fixed before any ordering takes place. The unanimity relation $\succ_u$ can be computed without relying on the execution order. There is no circularity.
 
The price is that a transaction may declare that it could touch both applications $a_t$ and $a_s$, yet in the realized execution touch only $a_t$; the rule still treats it as interacting with both.  As a consequence, a transaction that is single-opinion in reality is treated as multi-opinion: it loses the unconditional protection of Proposition~\ref{thm:single-opinion} and becomes eligible for demotion (Lemma~\ref{lem:two-multi-app}). Over-declaration thus enlarges the multi-opinion set, which is the rule's attack surface. The guarantees remain valid; they simply cover fewer transactions than the true interaction structure would allow. Also, a user that declares from the outset that its transaction will interact with a single opinionated application continues to enjoy the protection of Proposition~\ref{thm:single-opinion}. The declaration must cover validity as well: whether an interaction is valid can depend on state --- a cancellation is valid until the order fills --- so the declared reading treats a transaction as interacting with every application it may validly interact with, with the same over-declaration cost as before.
 
\paragraph{Actual interactions.} The second option is to build $X_t$ from each transaction's \emph{actual} behavior. This is tighter --- only genuinely multi-opinion transactions are exposed --- but actual interactions depend on the order, so the definition is circular. The circularity is a fixed-point condition: a consistent order is one that reproduces the interaction sets it was derived from.
 
This raises three questions the declared reading sidesteps. \emph{Existence}: a consistent order need not obviously exist. \emph{Uniqueness}: if several orders are each self-consistent, the protocol must still choose one, which reintroduces the consensus problem that delegating the order to a fixed rule was meant to remove. \emph{Computation}: the proposer must find a fixed point, and validators must verify it, cheaply enough to be practical. The endogenous definition delivers the strongest guarantees, but it puts the burden of a well-defined, unique, and verifiable order back on the protocol.
 
The choice between the two readings is a trade-off. Declared interactions give a unique, easily verified order and weaker guarantees; actual interactions give the strongest guarantees but must first settle existence, uniqueness, and computation of the fixed point. I leave a full treatment of the fixed-point problem --- conditions for existence and uniqueness, and a procedure validators can check --- for future work.
 
\section{Conclusion}
 
I have proposed a single rule, application unanimity override, that lets applications express ordering preferences and respects them whenever doing so is safe. Three guarantees hold against any attacker. Every strict comparison
involving a single-opinion transaction and another transaction at its sole
opinionated application is preserved, whatever the attacker does
(Proposition~\ref{thm:single-opinion}). A transaction in the gated region
preserves every strict unanimous comparison in which it ranks above another
transaction --- even when it is composable across several applications
(Proposition~\ref{thm:protected-region}) --- and a symmetric guarantee
protects transactions at the bottom of the order
(Proposition~\ref{prop:bottom-gated}). The guarantees cannot be improved upon: no ordering rule, whatever its
form, can guarantee an application's preferences over multi-opinion
transactions that are gated neither above nor below (Proposition~\ref{prop:necessity}). None of
the guarantees requires isolating the application on its own chain.
 
Consider again the applications that motivated the rule. A PropAMM's oracle update is single-opinion and gated to a designated address: it lies in the gated region and executes before every swap, which is the ordering the PropAMM needs. A cancellation on a central-limit-order-book exchange is the composable case. It touches many markets at once, yet if each market ranks cancellations first and accepts them only from the order owner, the cancellation is protected in all of them. In both cases the rule delivers the ordering the application was built around.
 
Auctions are the most instructive case, because there the guarantee does not come from the rule alone. An on-chain auction wants its bids ordered from highest to lowest, and the auction itself expresses that preference. Two bids that touch only the auction keep their correct order under any attack, so the highest bid wins. But this holds only while the bids stay single-opinion. A bidder who routes a bid through another opinionated application turns it into a multi-opinion transaction, which an attacker can try to demote by creating a cycle. Bidders therefore have every reason to keep their bids simple, and the auction's correctness rests on the rule and that incentive together.
 
Unanimity override does not enforce every preference an application might hold. It enforces the preferences that survive the threat model, and it tells senders which transactions are safe to send. Participants are then led, by their own incentives, to send the transactions it can protect.

Several questions remain open. In particular, the demotion rule studied here selects the lowest-priority eligible transaction; rules that exploit the topology of the cyclic component may narrow the attacker's options further. Also, defining the interaction sets by actual rather than declared interactions poses a fixed-point problem --- existence, uniqueness, and cheap verification of a consistent order --- that this paper only states.

\appendix
 
\section{Mathematical appendix}\label{app:math}

\subsection{Proof of Lemma~\ref{lem:termination}}
\label{app:proof-lem-termination}
 
\begin{proof}
Fix a cyclic component $Y$. At any iteration, if the relation induced by
$Y\setminus D$ contains a cycle, Lemma~\ref{lem:two-multi-app} implies that this cycle contains a multi-opinion transaction. Since the cycle lies in $Y\setminus D$, that transaction has not yet been demoted. Hence the set from which the fallback selects is nonempty.
 
Each iteration adds one previously undemoted transaction to $D$. Since $Y$ is finite, the selection procedure terminates. In particular, after all multi-opinion transactions in $Y$ have been demoted, no cycle can remain in the relation induced by $Y\setminus D$, again by Lemma~\ref{lem:two-multi-app}.
 
Let $K=Y\setminus D$ when the procedure stops. By construction, the relation
induced by $K$ is acyclic. The relation induced by $D$ is the strict total order
determined by $U_d$ and the protocol's deterministic tie-breaker, and is
therefore acyclic. Every repaired comparison between $K$ and $D$ points from
$K$ to $D$. Hence a cycle in the repaired graph induced by $Y$ could not use
vertices from both sets: after entering $D$, it could not return to $K$. Since
neither $K$ nor $D$ contains a cycle internally, the repaired graph induced by
$Y$ is acyclic.
 
The above shows that there is no cycle left within each cyclic component. To conclude the proof, note that demotion does not affect comparisons between transactions belonging to different cyclic components, and every comparison it imposes stays within a single component. By definition of a cyclic component, there cannot be a cycle spanning multiple cyclic components (else there would be a single cyclic component instead of separate cyclic components).  Since the repaired relation is acyclic within each component and no cycle crosses components, it is acyclic.
\end{proof}

\subsection{Proof of Lemma~\ref{lem:preservation}}
\label{app:proof-lem-preservation}

\begin{proof}
The repair can remove a comparison $x\succ_u y$ only when $x$ and $y$ lie in the same cyclic component $Y$ and at least one of them is demoted. Since $x$ is not demoted, this can happen only if $y\in D$ and $x\in K=Y\setminus D$; in that case demotion imposes $x\mathrel{\tilde\succ_u}y$. In every other case the comparison is left unchanged. Hence $x\mathrel{\tilde\succ_u}y$.
\end{proof}
 
\subsection{Proof of Proposition~\ref{prop:single-opinion-anchor}}
\label{app:proof-prop-single-opinion-anchor}
 
\begin{proof}
Fix $y\in X_t\setminus\{x\}$. Since $x$ is single-opinion, every application $a_s$ shared by $x$ and $y$ with $s\neq t$ is non-opinionated. Hence the unanimity ranking between $x$ and $y$ coincides with the preferences of application $a_t$.  If $U_t(x)>U_t(y)$, then  $x\succ_u y$. Because the fallback never demotes the single-opinion transaction $x$,  Lemma~\ref{lem:preservation} implies $x\mathrel{\tilde\succ_u}y$.  
If $U_t(y)>U_t(x)$, the symmetric argument gives $y\succ_u x$. If $y$ is not demoted, Lemma~\ref{lem:preservation} yields $y\mathrel{\tilde\succ_u}x$.
 
Finally, if both $x$ and $y$ are single-opinion transactions whose sole opinionated application is $a_t$, neither can be demoted, so whichever strict comparison $a_t$ makes is preserved.
\end{proof}

\subsection{Proof of Proposition~\ref{thm:protected-region}}
\label{app:proof-prop-top-transaction}

\begin{proof}
Fix a feasible strategy $\sigma$.

\emph{Step 1: no member of $X^*$ is demoted.} Suppose the fallback demotes
$w\in X^*$. At that iteration, $w$ lies on a directed cycle $Q$ of the
residual unanimity relation. No comparison of $Q$ can enter $X^*$ from
outside: it would place a transaction in $X^\sigma\setminus X^*$ unanimously
above a member of $X^*$, contradicting unreachability. Following $Q$
backward from $w$ therefore shows, one predecessor at a time, that every
vertex of $Q$ lies in $X^*$; then $Q$ is a cycle in the restriction of
$\succ_u^\sigma$ to $X^*$, contradicting acyclicity.

\emph{Step 2: preservation under attack.} Fix $x^*\in X^*$ and
$w\in X^\sigma$ with $x^*\succ_u^\sigma w$. By Step~1, $x^*$ is not
demoted, so Lemma~\ref{lem:preservation} gives
$x^*\mathrel{\tilde\succ_u^\sigma}w$, and Kahn's algorithm extends the
repaired relation: $x^*\succ_e^\sigma w$.

\emph{Step 3: no shared application ranks an attacker transaction above a
member.} Fix $x^*\in X^*$, an attacker transaction $z\in R$, and
$t\in T^{\mathrm o}(x^*)\cap T^{\mathrm o}(z)$, and suppose
$U_t(z)>U_t(x^*)$. By (R1), $R$ contains a transaction $z_t$ with
$T^{\mathrm o}(z_t)=\{t\}$ and $U_t(z_t)=U_t(z)$. Under the feasible
strategy that injects $z_t$ as a fresh transaction, $a_t$ is the only
opinionated application ranking the pair $(z_t,x^*)$ and strictly prefers
$z_t$, so $z_t\succ_u x^*$. Since $z_t$ is fresh, it lies outside $X$ and
hence outside $X^*$, contradicting unreachability. Hence
$U_t(x^*)\geq U_t(z)$.

\emph{Step 4: the guarantee.} Fix $z\in X$ with $x^*\succ_u z$, and write
$\bar z\equiv\kappa^\sigma(z)$. If $z$ is not replaced, then $\bar z=z$:
neither transaction of the pair is replaced, and the attacker can change
neither their interactions nor the honest applications' preferences, so the
baseline comparison persists, $x^*\succ_u^\sigma z$. If $z$ is replaced,
$\bar z$ extends $z$, so every application shared by $x^*$ and $z$ ranks
$\bar z$ exactly as it ranked $z$: all weakly prefer $x^*$, and at least
one strictly. Every opinionated application shared by $x^*$ and $\bar z$
but not by $x^*$ and $z$ satisfies $U_t(x^*)\geq U_t(\bar z)$ by Step~3,
since $\bar z\in R$; every non-opinionated one is indifferent. Hence
$x^*\succ_u^\sigma\bar z$, and Step~2 gives
$x^*\succ_e^\sigma\kappa^\sigma(z)$. In particular, no target pair whose
higher-ranked transaction lies in $X^*$ can be successfully attacked,
through either channel.
\end{proof}

\subsection{Proof of Proposition~\ref{prop:bottom-gated}}
\label{app:proof-prop-bottom-gated}

\begin{proof}
The argument mirrors the proof of Proposition~\ref{thm:protected-region},
with the direction of every comparison reversed; I record the differences.
Fix a feasible strategy $\sigma$ and $x_*\in X_*$.

\emph{Step 1: no comparison involving $x_*$ is removed or imposed.} The
transaction $x_*$ lies on no directed cycle of $\succ_u^\sigma$: no
comparison of such a cycle could leave $X_*$, since it would rank a member
of $X_*$ above a transaction outside it; following the cycle forward from
$x_*$ would then keep it inside $X_*$, contradicting acyclicity. A
transaction on no cycle belongs to no cyclic component, so the fallback
removes no comparison involving $x_*$, and every comparison it imposes
joins two members of a cyclic component, so none involves $x_*$.

\emph{Step 2.} As before: if $w\succ_u^\sigma x_*$, the comparison
survives to the repaired relation by Step~1, and Kahn's algorithm extends
it: $w\succ_e^\sigma x_*$.

\emph{Step 3.} Step~3 of the previous proof with the inequality reversed:
no application shared by $x_*$ and an attacker transaction $w\in R$
satisfies $U_t(x_*)>U_t(w)$. Otherwise, injecting the single-opinion
decomposition $w_t$ of $w$ at $a_t$ would make the member $x_*$ of $X_*$
unanimously ranked above a fresh transaction outside $X_*$, contradicting
condition (i). Hence $U_t(w)\geq U_t(x_*)$.

\emph{Step 4.} Identical to Step~4 of the previous proof up to direction.
Fix $z\in X$ with $z\succ_u x_*$ and write $\bar z\equiv\kappa^\sigma(z)$:
every application shared by the pair in the baseline ranks $\bar z$ as it
ranked $z$, and every newly shared opinionated application weakly prefers
$\bar z$ by Step~3, so $\bar z\succ_u^\sigma x_*$; Step~2 then gives
$\kappa^\sigma(z)\succ_e^\sigma x_*$.
\end{proof}

\subsection{Proof of Proposition~\ref{thm:single-opinion}}
\label{app:proof-lem-same-app-invariance}

\begin{proof}
Fix a feasible attacker strategy $\sigma$ with $x_i\in X^\sigma$ and a transaction $z\in X^\sigma\cap\mathbb X_t$ with $U_t(x_i)>U_t(z)$. A baseline transaction appears in the attacked block only unchanged, so $x_i$ retains its baseline interactions, and a newly deployed application's domain contains only transactions that interact with it: $x_i$ remains single-opinion at $a_t$ in the attacked profile. Proposition~\ref{prop:single-opinion-anchor} then gives $x_i\mathrel{\tilde\succ_u^\sigma}z$, and Kahn's algorithm extends the repaired relation: $x_i\succ_e^\sigma z$.

For the target-pair statement, fix a target pair $(x_i,x_j)$ and a feasible strategy $\sigma$. A successful attack requires $\kappa^\sigma(x_j)\succ_e^\sigma x_i$, and therefore $x_i\in X^\sigma$. Write $\bar x_j=\kappa^\sigma(x_j)$. Because $x_i$ is single-opinion and $x_i\succ_u x_j$, its sole opinionated application $a_t$ is shared by the baseline pair and strictly prefers $x_i$: $U_t(x_i)>U_t(x_j)$. If $x_j$ is replaced, the extension ranks identically at every application $x_j$ interacts with, so $U_t(\bar x_j)=U_t(x_j)$; if it is not replaced, the equality is immediate. Hence $U_t(x_i)>U_t(\bar x_j)$, and the first part, applied with $z=\bar x_j$, yields $x_i\succ_e^\sigma\bar x_j$ --- excluding $\bar x_j\succ_e^\sigma x_i$. Since $\sigma$ was arbitrary, no feasible strategy is a successful attack on the pair.
\end{proof}

\subsection{Proof of Proposition~\ref{prop:necessity}}
\label{app:proof-prop-necessity}

\begin{proof}
Let $x$ be a possible transaction that is multi-opinion and gated neither
above nor below.

\emph{Step 1: opposing witnesses at distinct applications.} I claim there
exist distinct $t_1,t_2\in T^{\mathrm o}(x)$ and possible transactions
$y_1,y_2$ with $U_{t_1}(y_1)>U_{t_1}(x)$ and $U_{t_2}(y_2)<U_{t_2}(x)$.
First, every $t\in T^{\mathrm o}(x)$ admits a witness above or below $x$:
because $a_t$ is opinionated, $U_t$ is not constant on $\mathbb X_t$, so
some possible transaction $z$ satisfies $U_t(z)\neq U_t(x)$, and $z$ ranks
strictly above or strictly below $x$ at $a_t$. Now, since $x$ is not gated
above, some $t_1\in T^{\mathrm o}(x)$ admits a witness above $x$; since
$x$ is not gated below, some $t_2\in T^{\mathrm o}(x)$ admits a witness
below $x$. If $t_1\neq t_2$, the claim holds. If $t_1=t_2$, pick any other
$t'\in T^{\mathrm o}(x)$, which exists because $x$ is multi-opinion: $t'$
admits a witness, above or below $x$. If above, replace $t_1$ by $t'$; if
below, replace $t_2$ by $t'$. Either way the two applications are
distinct.

\emph{Step 2: the conflicting transaction.} By single-opinion
decomposition, applied to any sender, take the decomposition $y_1'$ of
$y_1$ at $a_{t_1}$ and the decomposition $y_2'$ of $y_2$ at $a_{t_2}$,
each interacting with its application alone; then
$U_{t_1}(y_1')=U_{t_1}(y_1)>U_{t_1}(x)$ and
$U_{t_2}(y_2')=U_{t_2}(y_2)<U_{t_2}(x)$. Since $t_1\neq t_2$, the two
decompositions interact with disjoint sets of applications, so by bundling
there is a possible transaction $y$ that $a_{t_1}$ ranks as it ranks
$y_1'$ and $a_{t_2}$ ranks as it ranks $y_2'$:
\[
  U_{t_1}(y)>U_{t_1}(x),
  \qquad
  U_{t_2}(y)<U_{t_2}(x).
\]
Setting $t=t_2$ and $t'=t_1$ delivers the two inequalities of the
statement, and $y\neq x$ follows from $U_{t_1}(y)>U_{t_1}(x)$. It remains
to check that $y$ belongs to the class. The transaction $y$ interacts with
the opinionated applications $a_{t_1}$ and $a_{t_2}$, so it is
multi-opinion. It is not gated above: $x$ ranks strictly above $y$ at
$a_{t_2}$. It is not gated below: $x$ ranks strictly below $y$ at
$a_{t_1}$.
\end{proof}

\subsection{Proof of Proposition \ref{prop:repaired-phi-extends-repaired-unanimity}}\label{last-proof}
    \begin{proof}
Since $\Phi$ is Paretian, Proposition~\ref{prop:paretian-extension} gives $\succ_u\subseteq\succ_\Phi$. The first stage of the
repair deletes only non-unanimity comparisons, so every unanimity comparison
remains present. Let $\succ_\Phi'$ denote the relation after this first stage.
 
By construction, no cycle of $\succ_\Phi'$ contains a non-unanimity comparison.
Therefore every remaining cycle of $\succ_\Phi'$ is a cycle of $\succ_u$.
Equivalently, every cyclic component of $\succ_\Phi'$ is a cyclic component of
$\succ_u$: if a non-unanimity comparison belonged to a cyclic component, it
would lie on a directed cycle, contradicting the stopping condition of the first
stage.
 
The second stage therefore applies the demotion fallback to exactly the same
cyclic components and the same directed cycles as the fallback applied to
$\succ_u$. With the same default order and the same deterministic tie-breaking,
it selects the same demoted set in each such component. Hence, inside every
cyclic component of $\succ_u$, the repair of $\succ_\Phi'$ produces the same
surviving and imposed comparisons as the repair of $\succ_u$. Outside those
components, all unanimity comparisons are left unchanged. Therefore every
comparison in $\tilde\succ_u$ is also present in $\tilde\succ_\Phi$, proving
\[
  \tilde\succ_u\subseteq\tilde\succ_\Phi .
\]
 
It remains to check that both stages terminate and that $\tilde\succ_\Phi$ is acyclic. Termination is immediate: the first stage deletes one comparison per step from a finite relation, and the second stage terminates by Lemma~\ref{lem:termination}. For acyclicity, I need to show that the additional comparisons imposed by demotion cannot form cycles with surviving non-unanimity comparisons.\footnote{Remember that the rule deletes non-unanimity comparisons \textit{that lie on a cycle} before demoting multi-opinion transactions, so some non-unanimity comparisons may still be present at the demotion stage.} The intuition is that cyclic components are strongly connected, so every imposed comparison connects two transactions that a directed path of unanimity comparisons already connects, in the same direction. In other words, demotion adds comparisons but no new reachability within the component. So if no cycle of $\succ_\Phi'$ involves a surviving non-unanimity comparison, such cycles remain absent after the demotion stage. 

%More precisely, recall from the previous paragraph that $\tilde\succ_\Phi$ consists of $\tilde\succ_u$, which is acyclic by Lemma~\ref{lem:termination}, plus the surviving non-unanimity comparisons. The stopping rule of the first stage guarantees that no cycle of $\succ_\Phi'$ contains a non-unanimity comparison. This guarantee does not extend automatically to $\tilde\succ_\Phi$: demotion does not only delete comparisons, it also imposes new ones, which did not exist when the first stage stopped. A cycle of $\tilde\succ_\Phi$ could therefore, in principle, combine surviving non-unanimity comparisons with imposed comparisons.

%Note, however,  demotion operates within a single cyclic component of $\succ_u$. Cyclic components are strongly connected, so there is a  directed path of unanimity comparisons between any two transactions on the same cyclic component, pointing in the same direction as the imposed comparison. In other words, demotion adds comparisons but no new reachability within the component, so the first stage has already examined every connection the imposed comparisons create. And because the first stage deletes only non-unanimity comparisons, every comparison on this path is present in $\succ_\Phi'$ --- whether or not it later survives demotion. Replacing each imposed comparison on the hypothetical cycle with its unanimity path therefore yields a cycle of $\succ_\Phi'$ that contains the same non-unanimity comparison, which the stopping rule excludes. Every cycle of $\tilde\succ_\Phi$ must then lie entirely in $\tilde\succ_u$, which is acyclic.

More precisely, let
\[
  N' \equiv \succ_\Phi'\setminus \succ_u
\]
be the set of non-unanimity comparisons that survive the first stage. Thus
\[
  \tilde\succ_\Phi
  =
  \tilde\succ_u \cup N',
\]
where $\tilde\succ_u$ is acyclic by Lemma~\ref{lem:termination}. It remains
only to rule out cycles that use at least one comparison from $N'$. The
comparisons of $\tilde\succ_u$ are of two kinds: surviving unanimity
comparisons, which belong to $\succ_u$, and comparisons imposed by
demotion. Every imposed comparison joins two transactions in the same
cyclic component of $\succ_u$, which is strongly connected; there is
therefore a directed path of unanimity comparisons with the same endpoints
and the same direction. Now suppose a directed cycle of $\tilde\succ_\Phi$
contains a comparison from $N'$. Replace each imposed comparison on the
cycle by such a unanimity path. Every comparison of the resulting closed
walk lies in $N'$ or in $\succ_u$; and because the first stage deletes
only non-unanimity comparisons, every unanimity comparison is present in
$\succ_\Phi'$. The walk is therefore a directed cycle of $\succ_\Phi'$
that still contains the comparison from $N'$ --- contradicting the
stopping rule of the first stage, which halts only when no cycle of
$\succ_\Phi'$ contains a non-unanimity comparison. Every directed cycle of
$\tilde\succ_\Phi$ must then lie entirely in $\tilde\succ_u$, which is
acyclic. Hence $\tilde\succ_\Phi$ is acyclic.
\end{proof}

\bibliographystyle{plainnat}
\bibliography{biblio}
 
\end{document}